\documentclass[runningheads]{llncs}

\usepackage{graphicx}

\usepackage{cite}
\usepackage{amsmath,amsfonts,amssymb,amscd}
\usepackage{balance}
\usepackage{mathtools}
\usepackage{algorithm2e}
\usepackage{textcomp}
\usepackage{xcolor}
\usepackage{caption}
\usepackage{subcaption}
\usepackage{siunitx}
\usepackage{stmaryrd}
\usepackage{aligned-overset}
\usepackage{bbm}
\usepackage{breqn}
\usepackage{verbatim}
\usepackage{todonotes}
\usepackage{tabularx}
\usepackage{colortbl}
\usepackage{hhline} 
\usepackage{svg}
\usepackage{appendix}
\usepackage{float}
\usepackage{apxproof}
\usepackage{xcolor}

\newcommand{\topt}{\theta_{\mathrm{opt}}}

\newcommand{\up}[1]{\overline{#1}}

\newcommand{\cross}{\alpha_2}
\newcommand{\rl}{\beta_{R,T}}
\newcommand{\loopt}{\beta_{\theta_\mathrm{opt}}^1}
\newcommand{\lo}{\beta_{\theta}^1}

\newcommand{\foi}{\alpha_1}

\newcommand{\post}{\tilde{\Theta}}
\newcommand{\diff}{\Theta_{\mathrm{diff}}}
\newcommand{\init}{\Theta}
\newcommand{\qmin}{q_{\mathrm{min}}}
\newcommand{\trel}{\mathcal{T}_{\beta_\theta^1}^{\mathrm{rel}}}
\newcommand{\tabs}{\mathcal{T}_{\beta_{\theta}^1}^{\mathrm{abs}}}
\newcommand{\mt}{\mathcal{T}}
\newcommand{\discot}{\theta_{\mathrm{disco}}}
\newcommand{\astar}{\tau}
\newcommand{\muex}{\mu_{\mathrm{ex}}}
\newcommand{\muheu}{\mu_{\mathrm{heu}}}
\newcommand{\execex}{t_{\mathrm{ex}}^{\mathrm{exec}}}
\newcommand{\execheu}{t_{\mathrm{heu}}^{\mathrm{exec}}}
\newcommand{\toptp}{\tilde{\theta}_{\mathrm{opt}}}

\newcommand{\gb}[1]{\SI{#1}{\giga\byte}}
\newcommand{\mhz}[1]{\SI{#1}{\mega\hertz}}
\newcommand{\ghz}[1]{\SI{#1}{\giga\hertz}}

\newcommand{\s}[1]{\SI{#1}{\second}}

\newcommand{\fn}{\mathcal{F}^\uparrow_0}
\newcommand{\f}{\mathcal{F}}
\newcommand{\fa}{\mathcal{F}^\uparrow}

\newcommand{\posPart}[1]{\left[#1\right]^+}
\newcommand{\verDev}[2]{v \! \left(#1, #2\right)}
\newcommand{\horDev}[2]{h \! \left(#1, #2\right)}

\newtheorem{thm}{Theorem}

\newtheorem{defn}[thm]{Definition}
\newtheorem{rem}[thm]{Remark}
\newtheorem{lem}[thm]{Lemma}

\DeclareMathOperator*{\argmin}{arg\,min}

\allowdisplaybreaks

\begin{document}

\title{Minimal Per-Flow Backlog Bounds \\ at an Aggregate FIFO Server \\under Piecewise-Linear Arrival Curves}

\titlerunning{Minimal FIFO Per-Flow Backlog Bounds under PWL Arrival Curves}

\author{Lukas Wildberger \and Anja Hamscher \and Jens B. Schmitt}

\authorrunning{L. Wildberger et al.}

\institute{DISCO Lab, RPTU Kaiserslautern-Landau, 67663 Kaiserslautern, Germany \\ 
\email{lukas.wildberger@cs.rptu.de}, \email{\{hamscher,jschmitt\}@cs.uni-kl.de}}

\maketitle           

\begin{abstract}
Network Calculus (NC) is a versatile methodology based on min-plus algebra to derive worst-case \emph{per-flow} performance bounds in networked systems with many concurrent flows. In particular, NC can analyze many scheduling disciplines; yet, somewhat surprisingly, an aggregate FIFO server is a notoriously hard case due to its min-plus \emph{non-linearity}. A resort is to represent the FIFO residual service by a family of functions with a free parameter instead of just a single curve. For simple token-bucket arrival curves, literature provides optimal choices for that free parameter to minimize delay and backlog bounds.
In this paper, we tackle the challenge of more general arrival curves than just token buckets. In particular, we derive residual service curves resulting in minimal backlog bounds for general piecewise-linear arrival curves. To that end, we first show that a backlog bound can always be calculated at a breakpoint of either the arrival curve of the flow of interest or its residual service curve. Further, we define a set of curves that characterize the backlog for a fixed breakpoint, depending on the free parameter of the residual service curve. We show that the backlog-minimizing residual service curve family parameter corresponds to the largest intersection of those curves with the arrival curve. In more complex scenarios finding this largest intersection can become inefficient as the search space grows in the number of flows. Therefore, we present an efficient heuristic that finds, in many cases, the optimal parameter or at least a close conservative approximation. This heuristic is evaluated in terms of accuracy and execution time. 
Finally, we utilize these backlog-minimizing residual service curves to enhance the DiscoDNC tool and observe considerable reductions in the corresponding backlog bounds.


\end{abstract}

\section{Introduction}

First-In First-Out (FIFO) is a popular scheduling policy for networked systems due to its simplicity and low cost of implementation. In various network analysis methods, it is an interesting policy to analyze. One such analysis method is Network Calculus (NC) \cite{cruz1991calc1}, which is a versatile methodology for deriving performance bounds in networked systems \cite{cruz1998sced+, le2001network, blanc2006quality}. In particular, we are interested in the derivation of backlog bounds. It is straightforward to obtain a tight backlog bound at a FIFO node when considering all incoming traffic as a aggregate flow \cite{le2001network}, such as in switches with a shared queue. An issue arises when we are instead interested in a \emph{per-flow} backlog bound, as is the case for separate FIFO queues like in input-buffered switches with virtual output queues (VOQ) \cite{mckeown2002achieving}, or big data processing such as in Hadoop, where a job queue is fed data from a distributed file system through switches \cite{shvachko2010hadoop}.

With NC, per-flow backlog bounds can be calculated using a \emph{residual service curve}, which represents the leftover service available to a specific flow of interest (foi). However, due to the min-plus non-linearity of FIFO-scheduled systems \cite{liebeherr2009system}, computing a residual service curve is hard. 
Defining a \emph{family of functions} with a free parameter $\theta$, rather than a single FIFO residual service curve, is a way of dealing with this. The family of FIFO residual service curves \cite{cruz1998sced+} is given by 
\begin{equation*}
    \lo(t)=\left[\beta(t) - \alpha_2(t-\theta) \right]^+ \wedge \delta_{\theta}(t), \textrm{with } \theta \geq 0,
\end{equation*}
where $\beta$ models the service available to the traffic aggregate and $\cross$ is an upper bound on the traffic of all the cross flows that are multiplexed into the aggregate together with the foi. Since each $\lo$ results in a backlog bound, the question about the optimal value of $\theta$ to find a minimal backlog bound arises. In fact, a closed form for $\theta$ exists when only considering token-bucket-constrained arrivals \cite{le2001network, blanc2006quality}. Yet, for more general functions such as piecewise-linear (PWL) curves, to the best of our knowledge there exist no per-flow backlog bound results in literature regarding the calculation of $\theta$ values.
Somewhat the only exception is \cite{cholvi2002worst}, where an output bound for PWL arrival curves has been derived, yet without using NC. Nevertheless, it is interesting as the burst term of an output bound is also a bound on the backlog. The analysis is however restricted to the case of a constant rate server, whereas we deal with PWL service curves.

To that end, we make the following contributions in this paper:
\begin{itemize}
    \item In Section~\ref{sec:backlog_exact}, we derive an exact method to find the backlog-minimizing $\theta$ value for PWL-constrained arrival and service curves.
    \item An efficient heuristic that determines the backlog-minimizing parameter $\theta$ in most cases is presented in Section~\ref{sec:heuristic}. We evaluate its accuracy and execution time in comparison to the exact method.
    \item Finally, we show in Section~\ref{sec:dblanalysis} that the parameter $\theta$, derived using our new methods, poses a significant accuracy improvement for backlog bounds computed by the DiscoDNC tool over its current default setting.
\end{itemize}

\section{Network Calculus Background} \label{sec:nc_background}

Let $\mathbb{R}^+$ be the set of non-negative real numbers. $\f:=\{f: \mathbb{R}^+ \rightarrow \mathbb{R}\cup\{+\infty\}\}$ is the set of (min, plus) functions. Based on $\f$, we let $\fa$ be the set of non-decreasing functions $f\in\mathcal{F}$, and $\fn$ be the set of functions in $\fa$ with $f(0)=0$.

\begin{defn} [Basic Operators \cite{bouillard2018deterministic}] \label{def:operators}
    Let $f, g\in\f$. The min-plus convolution of $f$ and $g$ is defined as $f\otimes g(t)\coloneqq \inf_{0\leq s\leq t}\{f(t-s) + g(s)\}$. The (max-plus) deconvolution is defined as as $f \up{\oslash} g(t) \coloneqq \inf_{s \geq 0}\{f(t+s)-g(s)\}$.
\end{defn}

The \textit{impulse function} $\delta_T(t)$ is defined as $\delta_T(t)=\infty$, if $t > T$ and $0$ otherwise. The \textit{indicator function} $\mathbbm{1}_{A}$ is defined as $\mathbbm{1}_{A} = 1$, if $A$ is true and $0$ otherwise. 
For a given $\beta \in \mathcal{F}$, the \textit{lower non-decreasing closure} is defined as the largest non-decreasing function with $\beta_\downarrow \leq \beta$, given by $\beta_\downarrow \coloneqq \beta\up{\oslash}0$ \cite[p.~107]{bouillard2018deterministic}.

\begin{defn} [Pseudo-Inverse \cite{bouillard2018deterministic}] \label{def:pseudo_inverse}
    Let $f \in \mathcal{F}$ be a non-negative and non-decreasing function. The pseudo-inverse $f^{-1}$ is for all $x \in \mathbb{R}^+$ given by
    \begin{align*}
        f^{-1}(x) = \inf\{t \ | \ f(t) \geq x\}.
    \end{align*}
\end{defn}

Next, we define various notions that are used to model a network and derive its performance bounds. Let $A, D\in\fn$ be the \emph{cumulative arrival} and \emph{departure process} of a flow in the network, assuming causality $A\geq D$. Furthermore, we assume the system to be lossless. 
We define the most important performance measures for such a system:

\begin{defn} [Virtual Delay at Time $t$]
 	\label{def:delay}
 	The \emph{virtual delay} of data arriving at system $\mathcal{S}$ at time $t$ is the time until this data would be served, assuming FIFO-per-flow  order of service,
 	\begin{align}
 	    d_{A,D}(t) = \inf\{d \geq 0 : A(t) \leq D(t+d) \} = D^{-1}(A(t))-t \label{eq:virtual-delay}
 	\end{align}
\end{defn}

\begin{defn} [Backlog at Time $t$] \label{def:backlog}
	The \emph{backlog} of system $\mathcal{S}$ at time $t$ is the vertical
 	distance between arrival process $A$ and departure process $D$
 	at time $t$,
 	\begin{equation}
 	q_{A,D}(t) \coloneqq A(t)-D(t).
 	\label{eq:backlog}
 	\end{equation}
 \end{defn}

Arrival and service curves are central for the performance analysis using NC. 

\begin{defn} [Arrival Curve]
	\label{def:arrival-curve}
	Let $\alpha \in\fn$. We say that $\alpha$ is an \emph{arrival curve} for arrival process $A$ if it holds for all $0\leq s\leq t$ that
	\begin{equation*}
		A(t) - A(s) \leq \alpha(t-s)
		\Longleftrightarrow A = A \otimes \alpha. 
	\end{equation*}
\end{defn}
An example is a \emph{token-bucket}  arrival curve $ \gamma_{r,b}(t) = b+rt $ if $t > 0$, $\gamma_{r,b}(0)=0$. 

\begin{defn} [Service Curve] \label{def:service-curve}
    Let a flow with arrival process $A$ and departure process $D$ traverse a system $\mathcal{S}$. The system offers a \emph{min-plus service curve} $\beta$ to the flow if $\beta \in \f$ and it holds for all $t \geq 0$ that \vspace*{-1mm}
	\begin{equation*}
		D(t) \geq A \otimes \beta \ (t) = \inf_{0\leq s\leq t} \left\{A(t-s) + \beta(s)\right\}.
	\end{equation*}
\end{defn}\vspace*{-1mm}
An example is a \emph{rate-latency}  curve $ \beta_{R,T}(t) \coloneqq R\cdot \posPart{t-T}$,  $\posPart{x}\coloneqq \max\{x,0\}$.
We define two characteristic distances between functions.

\begin{defn} \label{def:horizontal_deviation}
    Let $f,g \in \mathcal{F}$. The horizontal deviation between $f$ and $g$ is
    \begin{align*}
        h(f,g)  &\coloneqq \sup_{t \geq 0} \{\inf \{d \geq 0 \ | \ f(t) \leq g(t+d) \}\} 
    \end{align*}
    and the \emph{vertical deviation} between $f$ and $g$ is
	\begin{equation*}
		\verDev{f}{g} \coloneqq \sup_{t\geq0}\left\{ f(t)-g(t)\right\}. 
	\end{equation*}
\end{defn}
Using these concepts, one can derive a backlog bound \cite[p.~115]{bouillard2018deterministic},\cite[p.~118]{le2001network}. 

\begin{thm} [Backlog Bound] \label{thm:performancebounds}
    Assume a single flow with arrival process $A$, with arrival curve $\alpha\in\fn$, and departure process $D$ traverses a system $\mathcal{S}$. Let the system $\mathcal{S}$ offer a service curve $\beta\in\fn$. 
    The backlog $q(t)$ satisfies for all t 
    \begin{equation*}\label{eq:backlogbound}
        q_{A,D}(t)\leq \verDev{\alpha}{\beta}. 
    \end{equation*}
\end{thm}
A FIFO residual service curve can be calculated as follows.
\begin{thm} [Residual Service Curve for FIFO \cite{cruz1998sced+}] \label{thm:FIFO_leftover_service_curve}
    Let $t \geq 0$. Consider a system $\mathcal{S}$ that multiplexes two flows $f_1$ and $f_2$ using FIFO scheduling. The arrivals of $f_2$, $A_2$, are constrained by $\alpha_2$. Further, assume that $\mathcal{S}$ guarantees a service curve $\beta$ to the aggregate of the flows. Then, for any $\theta \geq 0$, the residual service of $f_1$ is
    \begin{align} \label{Eq:FIFO_leftover_service_curve}
        \beta_{\theta}^1(t) &= [\beta(t) - \alpha_2(t-\theta)]^+ \wedge \delta_{\theta}(t) 
    \end{align}
\end{thm}

\begin{defn} [PWL Concave Normal Form \cite{bouillard2018deterministic}] \label{def:concave_piecewise_linear_normal_form}
    Let $r_i$, $b_i$ $\in \mathbb{R}^+$ and set $\gamma_i=\gamma_{r_i,b_i}$. The piecewise linear concave function 
    $f = \min \{\gamma_i\}$
    is said to be in normal form, if $\gamma_i$ are sorted by a decreasing rate and no $\gamma_i$ can be removed without modifying the minimum:
    \begin{align}
        &i < j \Rightarrow r_i > r_j, \label{eq:concp1} \\
        &\forall i, \exists t > 0, \forall j \neq i, \gamma_i(t) < \gamma_j(t) \label{eq:concp2}.
    \end{align}
    If $f = \min\{\gamma_i\}, i \in \{1,\dots,n\}$ is in normal form, then there is a sequence of $a_{i}$ of respective intersections of the linear functions $\gamma_i$ and $\gamma_{i+1}$. These intersections, denoted by $a_i$, are also called \textit{breakpoints} of $f$.
\end{defn}

\begin{defn} [PWL Convex Normal Form \cite{bouillard2018deterministic}] \label{def:convex_piecewise_linear_normal_form}
    Let $R_i$, $T_i$ $\in \mathbb{R}^+$ and set $\beta_i=\beta_{R_i,T_i}$. The PWL convex function
    $f = \max \{\beta_i\}$
    is said to be in normal form, if $\beta_i$ are sorted by an increasing rate and no $\beta_i$ can be removed without modifying the maximum:
    \begin{align*}
        &i < j \Rightarrow R_i < R_j \ \land \ \forall i, \exists t > 0, \forall j \neq i, \beta_i(t) > \beta_j(t).
    \end{align*}
    If $f = \max\{\beta_i\}, i \in \{1,\dots,n\}$ is in normal form, then there is a sequence of $s_{i}$ of respective intersections of the linear functions $\beta_i$ and $\beta_{i+1}$. These intersections, denoted by $s_i$, are also called \textit{breakpoints} of $f$.
\end{defn}

\begin{defn} \label{def:linear_segment_at_time_t}
    Let the linear segment of a given PWL (concave or convex) function $f$ at time $t$ be called $f^t$, with $f^t = \beta_{R,T}$ or $f^t = \gamma_{r,b}$. The rate, $r$ or $R$, and the y-axis intercept $b$ or x-axis intercept $T$ of this linear segment is then also referred to as $r^t, R^t, b^t, T^t$. 
\end{defn}

\begin{defn} \label{def:a_star_and_s_star}
    Let $\alpha$ be a PWL concave curve in normal form and $\beta$ a PWL convex curve in normal form. Let $A$ and $B$ be the sets of breakpoints of $\alpha$ and $\beta$, respectively.
    Let $I_A$ be the set of all breakpoints mapped to $\alpha$, defined as $I_A \coloneqq A \ \cup \ \{ \alpha^{-1}(\beta(s)) \ | \ s \in B \}$.
    Then the first point in time for which the corresponding rate of $\alpha$ is less than or equal to the corresponding rate of $\beta$, mapped to $\alpha$, is called $a^*_{\alpha,\beta}$. We define $a^*_{\alpha,\beta}$ as follows: 
    \begin{align*}
        a^*_{\alpha,\beta} &\coloneqq \min\{ i \in I_A : r^i \leq R^{\beta^{-1}(\alpha(i))}\}.
    \end{align*}
\end{defn}

\section{System Model}\label{sec:sysmod}

In this section, we introduce the system setting considered in this paper. 
We also discuss the challenge associated with this setting and how to deal with it. 
\begin{figure}
    \centering
    \includegraphics{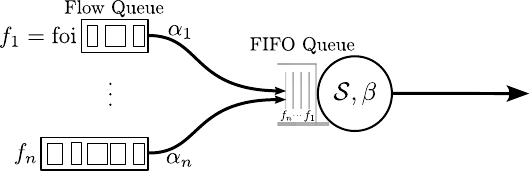}
    \caption{Distributed per-flow queues served by single FIFO system.}
    \label{fig:system_model}
\end{figure}

Fig.~\ref{fig:system_model} illustrates the system model under investigation.
Here, each (distributed) flow  (or job) $f_1, \ldots, f_n$ is associated with its own flow queue, each with its own size. 
These flows send their requests into a single (task) queue, which ensures the FIFO order between the individual flows when processed by system $\mathcal{S}$, for which we know a service curve $\beta$. 
The central question of our work now is: how to  appropriately size the individual flow queues, or, in other words, how to calculate per-flow backlog bounds.

We point out that per-flow queues in total generally require more space than an aggregate shared queue, i.e., we incur a \emph{segregation penalty}. The extent of that penalty needs to weighed against advantages from the distributed setting as given in Fig.~\ref{fig:system_model}. We come back to that issue in Section~\ref{sec:dblanalysis}.

Throughout the following sections, we assume that the flows $f_1 = \text{foi}, \ldots, f_n$ are constrained by PWL concave arrival curves $\alpha_1, \ldots, \alpha_n \in \mathcal{F}$. We aggregate all cross flows arrival curves $\alpha_2, \ldots , \alpha_n$, and simply call it $\alpha_2$. The system $\mathcal{S}$ that multiplexes the flows according to FIFO, offers a PWL convex service curve $\beta \in \mathcal{F}$ to the flow aggregate.

When applying NC to compute per-flow backlog bounds under PWL arrival and service curves, 
an issue arises: the residual service curve $\lo$ may contain a finite number of segments with negative slopes, which violates the non-decreasing property. While some literature formally requires that $\lo \in \fn$ \cite{le2001network}, this constraint is not always reflected in the conditions of the corresponding theorems. 
Anyway, the following lemma states that this violation does not impact the calculation of the vertical deviation (which we need to compute backlog bounds), since we can make use of the lower non-decreasing closure.

\begin{lem} 
    Let $\alpha \in \mathcal{F}$ and $\beta \in \mathcal{F}$ be given. Then it holds that
    \begin{align}
        v(\alpha,\beta) = v(\alpha, \beta_\downarrow). \label{Eq:vertical_dev_with_LNDC}
    \end{align}
\end{lem}

The proof for this and all subsequent lemmas can be found in the Appendix. 

\section{Derivation of the $\theta$ Parameter for Minimal Per-Flow Backlog Bounds}\label{sec:backlog_exact}

In this section, we derive the value $\topt$ that minimizes per-flow backlog bounds.
In order to derive such a value of $\theta$, we begin by formulating the problem mathematically as
\begin{align} \label{Eq:theta_opt_backlog_bound_main_goal}
    \topt = \argmin_{\theta \geq 0} \{ \verDev{\alpha_1}{\beta_{\theta}^1} \}.
\end{align}
In the following, we will transform this problem step by step until we are able to determine a solution. At first, in order to constrain the choice of $\theta$, we make use of the following lemma:

\begin{lem} \label{lem:theta<verDiv}
    Let $\alpha_1$ and $\alpha_2$ be PWL concave arrival curves and let $\beta$ be a PWL convex service curve under FIFO multiplexing. Let $h(\alpha_2,\beta)$ be the horizontal deviation of $\alpha_2$ and $\beta$. Then, it holds for $0 \leq \theta \leq h(\alpha_2,\beta)$ that
    \begin{align*}
        v(\alpha_1, \beta_{\theta}^1) \geq v(\alpha_1, \beta_{h(\alpha_2,\beta)}^1).
    \end{align*}
\end{lem}

Note that for token-bucket arrival curves and rate-latency service curves, the result of Lem.~\ref{lem:theta<verDiv} has already been stated in \cite{lenzini2005delay, bouillard2018deterministic, le2001network, blanc2006quality}.

Based on Lem.~\ref{lem:theta<verDiv}, we conclude that the search space for $\theta$ can be restricted to $\theta \geq \horDev{\alpha_2}{\beta}$. This allows us to transform the problem in Eq.~\eqref{Eq:theta_opt_backlog_bound_main_goal} to
\begin{align*}
    \topt = \argmin_{\theta \geq h(\alpha_2,\beta)} \{ \verDev{\alpha_1}{\beta_{\theta}^1} \}.
\end{align*}

To proceed, we define the family of functions $v_t(\theta)$ as the vertical distances between $\foi$ and $\lo$ at an arbitrary but fixed time $t$:
\begin{align*}
        v_t(\theta) \coloneqq
        &\begin{cases}
            \alpha_1(t) - \beta_{\theta}^1(t), & \hspace{4.7em} \mbox{if } h(\alpha_2,\beta) \leq \theta < t, \\ 
            \alpha_1(\theta), & \hspace{4.7em} \mbox{if } \theta \geq t > h(\alpha_2,\beta),
        \end{cases} \\
        = 
        &\begin{cases}
            \alpha_1(t) - \beta(t) + \alpha_2(t-\theta), & \mbox{if } h(\alpha_2,\beta) \leq \theta < t, \\ 
            \alpha_1(\theta), & \mbox{if } \theta \geq t > h(\alpha_2,\beta).
        \end{cases}
    \end{align*}
Here, we have used that $\beta_{\theta}^1(t) = \beta(t) + \alpha_2(t-\theta)$, for $h(\alpha_2,\beta) \leq \theta < t$. Particularly, by definition we have $\beta_{\theta}^1 = [\beta(t) - \alpha_2(t-\theta)]^+ \wedge \delta_{\theta}(t)$. 
Since for $\theta<t$ it holds that $\delta_{\theta}(t) = 0$
, and because $h(\alpha_2, \beta) \leq \theta$ ensures that $\beta(t) \geq \alpha_2(t-\theta)$ for all $t$ (according to Def.~\ref{def:horizontal_deviation}), both the positive part and $\delta_{\theta}(t)$ can be omitted. 

With this, the original problem in Eq.~\eqref{Eq:theta_opt_backlog_bound_main_goal} can now be rewritten as:
\begin{align}
    \topt &=  \argmin_{\theta \geq h(\alpha_2,\beta)} \{ \verDev{\alpha_1}{\beta_{\theta}^1} \} \notag \\
    &= \argmin_{\theta \geq h(\alpha_2,\beta)} \{ \sup_{t \geq 0} \{ \alpha_1(t) - \beta_{\theta}^1(t) \} \} \notag \\
    &= \argmin_{\theta \geq h(\alpha_2,\beta)} \{ \sup_{t \geq 0} \{v_t(\theta)\} \} \label{eq:problem_with_sup}
\end{align}

Here, the supremum is over all $t \geq 0$. 
The following lemma will be instrumental to restrict the relevant values of $t$ to a finite set. 

\begin{lem} \label{lem:pwl_backlog_bound_at_breakpoint}
    Let $\alpha$ be a PWL concave function and $\beta$ be a PWL convex function. Let $A$ and $B$ be the set of breakpoints of $\alpha$ and $\beta$, respectively. Then, the vertical deviation of $\alpha$ and $\beta$ can always be calculated at some time $t \in A \cup B$.
\end{lem}

\begin{rem} \label{rem:pwl_horizontal_deviation_at_breakpoint}
    The result of Lem.~\ref{lem:pwl_backlog_bound_at_breakpoint} also applies to the horizontal deviation.
\end{rem}

In the following, we denote the set of breakpoints of $\alpha_1, \alpha_2$, and $\beta$ as $A_1, A_2$, and $B$, respectively. 
Lem.~\ref{lem:pwl_backlog_bound_at_breakpoint} ensures that only $A_1$ and the set of breakpoints of $\beta_{\theta}^1$, induced by $A_2 \cup B$, need to be considered. A closer examination of these breakpoints reveals structural differences. In particular, the breakpoints of $\beta_{\theta}^1$ need further attention, as they depend on $\theta$. 
In contrast, the breakpoints of $\alpha_1$, $A_1$, remain invariant with respect to $\theta$. We call the set of breakpoints $A_1$ \emph{absolute time set} of $\alpha_1$. The absolute time set of the breakpoints of $\beta_{\theta}^1$, including the shift by $\theta$, is denoted by $\tabs$. 
The set of breakpoints of $\beta_{\theta}^1$, without being shifted by $\theta$, is called \textit{relative time set} and is given by
\begin{align*}
    \trel = A_2 \cup B.
\end{align*}
The relative time set $\trel$ allows us the usage of the breakpoints without having a dependency on $\theta$.
In the following, we derive from this relative time set an absolute time set of $\beta_{\theta}^1$, which is also no longer dependent on $\theta$.
This is achieved by determining a single value of $\theta$ for each $t \in \trel$. Formally, for each $t \in \mathcal{T}_{\beta_{\theta}^1}^{\text{rel}}$, the corresponding value $\theta_t$ is defined as the solution to the following equation:
\begin{align} 
    \alpha_1(\theta_t) = \alpha_1(t) - \beta_{\theta_t}^1(t). \label{Eq:matching_rule_rel_to_abs_time}
\end{align}

As established in Lem.~\ref{lem:pwl_backlog_bound_at_breakpoint}, only the breakpoints of the curves are relevant for computing the backlog bound.
So we consider each relative point in time $t \in \trel$, \emph{assuming} that it would be the one for which the vertical deviation is taken on. Then we obtain $\theta_t$ as the solution of Eq.~\eqref{Eq:matching_rule_rel_to_abs_time}, that minimizes this backlog bound under this assumption. 
In order to justify Eq.~\eqref{Eq:matching_rule_rel_to_abs_time}, let us consider an arbitrary, but fixed time $t \in \trel$. At this time $t$, the vertical distance between the two functions is given by $\alpha_1(t) - \beta_{\theta}^1(t)$, for $\theta < t$. This distance decreases as $\theta$ increases, so higher values of $\theta$ appear to be better for minimizing the backlog bound.
However, this perspective neglects an important factor — namely, the behavior of the vertical distance at time $\theta$. Specifically, as $\theta$ increases, the vertical distance at time $\theta$, given by $\alpha_1(\theta)$, also increases. Consequently, selecting a larger value of $\theta$ to reduce the distance at time  $t$ simultaneously results in a larger distance at time $\theta$. 
The optimal balance of this trade-off is given by $\theta_t$ as the solution of Eq.~\eqref{Eq:matching_rule_rel_to_abs_time}. 
For $\theta_t < t$, the respective $\theta_t$ can be derived as follows: 

\begin{align}
    \alpha_1(\theta_t) &= \alpha_1(t + \theta_t) - \beta^1_{\theta_t}(t + \theta_t) \notag \\
    \llap{$\Leftrightarrow$ \qquad \qquad \hspace{11.1em}} \alpha_1(\theta_t) &= \alpha_1(t + \theta_t) - \beta(t + \theta_t) + \alpha_2(t) \notag \\
    \llap{$\Leftrightarrow$ \qquad \qquad \hspace{0.6em}} \beta(t + \theta_t) - \alpha_1(t + \theta_t) + \alpha_1(\theta_t) &= \alpha_2(t) \notag \\
    \llap{$\Leftrightarrow$ \quad} \beta(t + \theta_t) - \alpha_1(t + \theta_t) + \alpha_1 \otimes \delta_t (t+\theta_t) &= \alpha_2(t) \notag \\
    \llap{$\Leftrightarrow$ \qquad \qquad \hspace{1.81em}} (\beta - \alpha_1 + (\alpha_1 \otimes \delta_t)) (t+\theta_t) &= \alpha_2(t), \label{eq:derive_theta_t}
\end{align}
\vspace*{-8mm}
\begin{align}
    \Rightarrow \theta_t &= d_{\beta - \alpha_1 + (\alpha_1 \otimes \delta_t), \alpha_2}(t) \label{eq:theta_t:line1} \\
    &= ([\beta - \alpha_1 + (\alpha_1 \otimes \delta_t)]^+)^{-1}(\alpha_2(t)) - t \label{eq:theta_t:line2}
\end{align}
where we used for Eq.~\eqref{eq:theta_t:line1} that the horizontal distance is given by the shift $\theta_t$ in Eq.~\eqref{eq:derive_theta_t} (see also Eq.~\eqref{eq:virtual-delay}). 
We apply the positive part in Eq.~\eqref{eq:theta_t:line2}, since the pseudo-inverse requires wide-sense increasing functions and the horizontal distance for a positive function $\alpha_2$ only requires non-negative functions.

Utilizing $\theta_t$, the absolute time set is now explicitly given as
\begin{align*}
    \tabs =
    \trel + \theta_t = \{ a_i^2 + \theta_{a_i^2} \ | \ a_i^2 \in A_2 \} \ \cup \ \{ s_i + \theta_{s_i} \ | \ s_i \in B \}.
\end{align*}
Combining this with $A_1$, we obtain the complete set of absolute breakpoint times:
\begin{align*}
    \mathcal{T} = A_1  \cup  \tabs.
\end{align*}

This reduces the number of $v_t(\theta)$ curves to be evaluated from an uncountable set (for $t \geq 0$) to a finite set (for $t \in \mathcal{T}$), which can be computed given the parameters of the arrival and service curves.
Since $t$ is now an element of the finite set $\mathcal{T}$, we can replace the supremum by a maximum in Eq.~\eqref{eq:problem_with_sup}, leading to the following transformed problem:
\begin{align*}
    \topt = \argmin_{\theta \geq h(\alpha_2,\beta)} \{ \max_{t \in \mathcal{T}} \{v_t(\theta)\} \}.
\end{align*} 

We can further restrict the domain of $\theta$. According to Lem.~\ref{lem:pwl_backlog_bound_at_breakpoint}, the backlog bound can always be calculated at a breakpoint of either $\alpha_1$ or $\beta_{\theta}^1$. Moreover, $\forall t \in \mathcal{T}$ and $\theta > t_{\max}$, with $t_{\max} = \max \mathcal{T}$, it holds that $v_t(\theta) = \alpha_1(\theta)$, which increases as $\theta$ increases. Hence, the backlog-minimizing value of $\theta$ cannot lie beyond $t_{\max}$.
Additionally, since $\max_{t \in \mathcal{T}} \{v_t(\theta)\} = \max \{ \max_{t \in \mathcal{T}} \{v_t(\theta)\}, \alpha_1(\theta) \}$ trivially holds, as $\alpha_1(\theta) = v_0(\theta) \in \{v_t(\theta) \ | \ t \in \mathcal{T} \}$, the problem can be transformed into
\begin{align}
    \topt= \argmin_{\theta \in [h(\alpha_2,\beta), t_{\max}]} \{ \max \{ \max_{t \in \mathcal{T}} \{v_t(\theta)\}, \alpha_1(\theta) \} \}. \label{eq:final_problem}
\end{align}

Fig.~\ref{fig:v_t_curves} gives a graphical representation of Eq.~\eqref{eq:final_problem}. Here, two $v_t$ curves, $v_{t_1}$ and $v_{t_2}$ are plotted against $\alpha_1(\theta)$. Each $v_t$ represents the vertical distance between $\alpha_1$ and $\lo$ at a specific breakpoint $t \in \mathcal{T}$. According to Lemma~\ref{lem:pwl_backlog_bound_at_breakpoint}, only these breakpoints need to be considered to compute backlog bounds.
As $\theta$ varies, the vertical distances $v_t(\theta)$ decrease at some of these breakpoints and increase at others. In order to minimize the overall backlog bound, we aim to find the value of $\theta$ that balances this trade-off such that the largest vertical distance, at all breakpoints, is minimized. 
Formally, this means that this (optimal) value of $\theta$ ($\topt$), marked by a red dot, is obtained at at the intersection of $\max_{t \in \mathcal{T}} \{v_t(\theta)\}$ and $\alpha_1(\theta)$.

\begin{figure}
    \centering
    \includegraphics[scale=0.58]{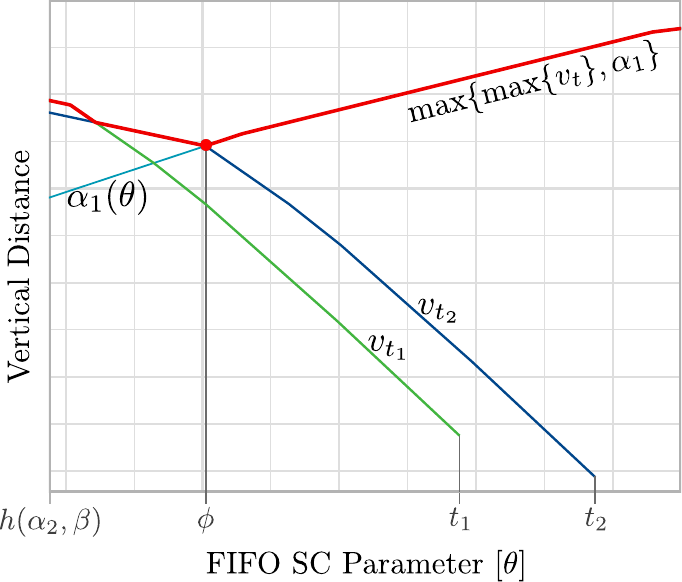}
    \caption{Vertical distances $v_t(\theta)$ and $\alpha_1(\theta)$, for $\theta \geq h(\alpha_2,\beta)$.}
    \label{fig:v_t_curves}
\end{figure}

It is useful to gain further insight into the problem, and examine the behavior of the two functions, $\max_{t \in \mathcal{T}} \{v_t(\theta)\}$ and $\alpha_1(\theta)$, at the boundary values of the interval for $\theta$. The following lemma characterizes the relation between these functions at the interval's endpoints.

\begin{lem} \label{lem:order_function_interval_bounds}
    Let $\max_{t \in \mathcal{T}} \{v_t(\theta)\}$ and $\alpha_1(\theta)$ be defined as above. For the endpoints of the given interval $[h(\alpha_2,\beta), t_{\max}]$ it holds that
    \begin{align*}
        \alpha_1(h(\alpha_2,\beta)) \leq \max_{t \in \mathcal{T}} \{v_t(h(\alpha_2,\beta))\} \ \textrm{and} \
        \alpha_1(t_{\max}) = \max_{t \in \mathcal{T}} \{v_t(t_{\max})\}.
    \end{align*}
\end{lem}

To address the problem given by Eq.~\eqref{eq:final_problem}, we introduce the concept of a \textit{first intersection} between $\max_{t \in \mathcal{T}} \{v_t(\theta)\}$ and $\alpha_1(\theta)$. 
This is essential for identifying the exact point in time where the function $\max_{t \in \mathcal{T}} \{v_t(\theta)\}$ drops onto the function $\alpha_1(\theta)$. 
We say that $\max_{t \in \mathcal{T}} \{v_t(\theta)\}$ and $\alpha_1(\theta)$ \textit{first intersect} at 
\begin{align*}
    \phi = \min \{ x \ | \ \alpha_1(x) = \max_{t \in \mathcal{T}} \{v_t(x)\} \}.
\end{align*}
It is important to note that a first intersection between $\max_{t \in \mathcal{T}} \{v_t(\theta)\}$ and $\alpha_1(\theta)$ always exists.
According to Lem.~\ref{lem:order_function_interval_bounds}, these two functions are guaranteed to intersect at least at $\theta = t_{\max}$.
Furthermore, by definition, every $v_t(\theta)$ curve, with $t \in \mathcal{T}$, is strictly decreasing for $\theta < t$ and strictly increasing for $\theta \geq t$. If $\max_{t \in \mathcal{T}} \{v_t(\theta)\}$ is taken at the respective strictly decreasing parts of the $v_t(\theta)$ curves, then the maximum function $\max_{t \in \mathcal{T}} \{v_t(\theta)\}$ itself is also strictly decreasing. The $\max_{t \in \mathcal{T}} \{v_t(\theta)\}$ is taken on the strictly decreasing parts until the maximum function $\max_{t \in \mathcal{T}} \{v_t(\theta)\}$ eventually equals $\alpha_1(\theta)$, which is strictly increasing.
This is exactly the first intersection $\phi$ of $\max_{t \in \mathcal{T}} \{v_t(\theta)\}$ and $\alpha_1(\theta)$, after which $\max_{t \in \mathcal{T}} \{v_t(\theta)\} = \alpha_1(\theta)$ holds. So we conclude that $\max_{t \in \mathcal{T}} \{v_t(\theta)\}$ is strictly decreasing for $\theta < \phi$ and strictly increasing for $\theta \geq \phi$. 
Note that if $\max_{t \in \mathcal{T}} \{v_t(\theta)\}$ is directly assumed to be on the strictly increasing part of the $v_t(\theta)$ curves, then the first intersection is exactly at $h(\alpha_2, \beta)$.

The following theorem provides the solution to the problem stated in Eq.~\eqref{eq:final_problem}, given by the first intersection $\phi$.

\begin{thm}\label{thm:argmin_of_max_is_opt_theta}
    Let $\max_{t \in \mathcal{T}} \{v_t(\theta)\}$ and $\alpha_1(\theta)$ be defined as above.
    Further, let the first intersection of $\max_{t \in \mathcal{T}} \{v_t(\theta)\}$ and $\alpha_1(\theta)$ be at $\phi$. Then, it holds
    \begin{align*}
        \topt= \argmin_{\theta \in [h(\alpha_2,\beta), t_{\max}]} \{ \max \{ \max_{t \in \mathcal{T}} \{v_t(\theta)\}, \alpha_1(\theta) \} \} = \phi.
    \end{align*}
    \begin{proof} \normalfont{
        According to Lem.~\ref{lem:order_function_interval_bounds} it holds that $\alpha_1(h(\alpha_2,\beta)) \leq \max_{t \in \mathcal{T}} \{v_t(h(\alpha_2,\beta))\}$ and  
        $\alpha_1(t_{\max}) = \max_{t \in \mathcal{T}} \{v_t(t_{\max})\}$ so the following also holds:
        \begin{align*}
            \forall \ \theta < \phi : \alpha_1(\theta) \leq \max_{t \in \mathcal{T}} \{v_t(\theta)\} \ \land \ 
            \forall \ \theta \geq \phi : \alpha_1(\theta) = \max_{t \in \mathcal{T}} \{v_t(\theta)\}.
        \end{align*}
        So the maximum function can be described as follows:
        \begin{align*}
            \max\{ \max_{t \in \mathcal{T}} \{v_t(\theta)\}, \alpha_1(\theta)\} = \begin{cases}
                \max_{t \in \mathcal{T}} \{v_t(\theta)\}, & \mbox{if } \theta < \phi, \\
                \alpha_1(\theta), & \mbox{if } \theta \geq \phi.
            \end{cases}
        \end{align*}
        
        Suppose $\argmin_{\theta \in [h(\alpha_2,\beta), t_{\max}]} \{ \max \{ \max_{t \in \mathcal{T}} \{v_t(\theta)\}, \alpha_1(\theta) \} \}$ is assumed at $\phi'$, with $ \phi' \in [h(\alpha_2,\beta), t_{\max}]$ and $\phi \neq \phi'$. So
        \begin{align*}
            \argmin_{\theta \in [h(\alpha_2,\beta), t_{\max}]} \{ \max \{ \max_{t \in \mathcal{T}} \{v_t(\theta)\}, \alpha_1(\theta) \} \} = \phi'
        \end{align*} should hold.
        Consider two cases: $\phi' < \phi$ and $\phi' > \phi$. \\
        \underline{Case I} ($\phi' < \phi$): Since $\phi' < \phi$ and $\max\{ \max_{t \in \mathcal{T}} \{v_t(\theta)\}, \alpha_1(\theta)\} = \max_{t \in \mathcal{T}} \{v_t(\theta)\}$ is strictly decreasing for $\theta < \phi$, it holds that: 
        \begin{align*}
            \max\{ \max_{t \in \mathcal{T}} \{v_t(\phi')\}, \alpha_1(\phi')\} 
            > \max\{ \max_{t \in \mathcal{T}} \{v_t(\phi)\}, \alpha_1(\phi)\}. \quad \lightning
        \end{align*} \\
        \underline{Case II} ($\phi' > \phi$): Since $\phi' > \phi$ and $\max\{ \max_{t \in \mathcal{T}} \{v_t(\theta)\}, \alpha_1(\theta)\} = \alpha_1(\theta)$ is strictly increasing for $\theta \geq \phi$, it holds that: 
        \begin{align*}
            \max \{ \max_{t \in \mathcal{T}} \{v_t(\phi')\}, \alpha_1(\phi') \} 
            > \max\{ \max_{t \in \mathcal{T}} \{v_t(\phi)\}, \alpha_1(\phi) \}. \quad \lightning
        \end{align*} 
        So $\topt= \argmin_{\theta \in [h(\alpha_2,\beta), t_{\max}]} \{ \max \{ \max_{t \in \mathcal{T}} \{v_t(\theta)\}, \alpha_1(\theta) \} \} = \phi$ holds.
    } \qed
    \end{proof}
\end{thm}

\begin{rem} \label{rem:first_intersections_is_last_intersection}
    Note that the first intersection of $\max_{t \in \mathcal{T}} \{v_t(\theta)\}$ and $\alpha_1(\theta)$ is also the last of all intersections of $v_t(\theta)$ curves, $t \in \mathcal{T}$, with $\alpha_1(\theta)$. See again Fig.~\ref{fig:v_t_curves}.
\end{rem}

\section{Efficient Calculation of Near-Optimal Backlog Bounds}\label{sec:heuristic}
In the following, we expand on the calculation of the FIFO SC parameter $\topt$ and its corresponding minimal backlog bound $\qmin$. We strive for an efficient calculation of backlog bounds that are still near to the minimal ones from the previous section using Thm.~\ref{thm:argmin_of_max_is_opt_theta}. In particular, according to Rem.~\ref{rem:first_intersections_is_last_intersection}, we calculate the $v_t$ curves of all $t\in\mt$, calculate their respective intersection points with $\foi$, then determine $\phi$ as the maximum of all intersections to obtain $\topt$ and $\qmin = v(\foi, \loopt)$, where $\qmin$ is measured at $\topt$. This constitutes an exact method, but requires the calculation of the intersection of \emph{all} possible breakpoints $t\in\mt$ of $\foi$ and $\cross$. For an efficient calculation of more complex scenarios (in particular with more flows and more segments per flow), the question arises whether we can exclude certain $t \in \mt$ in a first step to reduce the number of $v_t$ curves and intersections we have to calculate. We first describe the central idea: Instead of considering $\foi$ as a whole, we take a \emph{decomposition} approach. From Def.~\ref{def:concave_piecewise_linear_normal_form}, we know that $\foi$ is the minimum of $n$ token-bucket arrival curves $\gamma_i$. Hence, we can split $\foi$ into its token-bucket segments $S_i, 1\leq i \leq n$. For a token-bucket foi, we can identify a particular breakpoint $\astar\in\trel$ where the backlog is taken on and find a closed form for $\topt$ that minimizes its backlog $\qmin$. 

\begin{thm}\label{thm:basecaseNew}
    Let the foi $\alpha_1$ be a token-bucket arrival curve $\gamma_{r_1,b_1}$ and let the cross-flow $\alpha_2$ be a PWL concave arrival curve in normal form. Further, let $\beta$ be a  PWL convex service curve. The FIFO SC parameter that minimizes the backlog bound is given by:
    \begin{align*}
        \topt &= \argmin_{\theta \in [h(\alpha_2,\beta), t_{\max}]} \{ \max_{t \in \mathcal{T}} \left\{ v_t(\theta) \right\} \} \\
        &= h(\alpha_2+\gamma_{r_1,0}, \beta)= \frac{b^{\tau}_2+(r_2^{\tau}+r_1)\cdot \tau}{R^{\beta^{-1}({\alpha_2(\tau))}}} + T^{\beta^{-1}({\alpha_2(\tau))}} - \tau,
    \end{align*}
    
    with $\tau = a^*_{\alpha_2 + \gamma_{r_1,0},\beta}$.
    
    \begin{proof} \normalfont{
        According to Rem.~\ref{rem:first_intersections_is_last_intersection}, the result of Thm.~\ref{thm:argmin_of_max_is_opt_theta} also states that the optimal $\theta$ is given by the last of all intersections of $v_t(\theta)$ curves and $\alpha_1(\theta)$. In order to obtain the last intersection, let us consider the intersections of $\alpha_1(\theta)$ and $v_t(\theta)$, with $t \in \trel$.
        The optimal value $\theta_t$ for each $t \in \trel$ can be calculated as follows:
        \begin{align*}
            \alpha_1(\theta_t) &= \alpha_1(t + \theta_t) - \beta(t + \theta_t) + \alpha_2(t) \\
        \Leftrightarrow \hspace{5.55em}   r_1 \theta_t + b_1 &= r_1 t + r_1 \theta_t + b_1 - \beta(t + \theta_t) + \alpha_2(t) \\
        \Leftrightarrow \hspace{5.55em}   \beta(t + \theta_t) &= r_1 t + \alpha_2(t) \\
        \Leftrightarrow \hspace{8.5em}   \theta_t &= \beta^{-1}(r_1 t + \alpha_2(t)) - t \\
        \Leftrightarrow \hspace{8.5em}   \theta_t &= d_{\alpha_2 + \gamma_{r_1,0}}(t). 
        \end{align*}
        
        Now we need to find the time at which the backlog is maximal. As the backlog bound coincides with $\alpha_1(\topt)$ and $\alpha_1$ is increasing, we look for the maximal $\theta_t$:
        \begin{align*}
        \topt= \sup_{t \in \trel} \{d_{\alpha_2 + \gamma_{r_1,0}, \beta}(t) \} 
        \overset{( Rem.~\ref{rem:pwl_horizontal_deviation_at_breakpoint})}{=} \sup_{t\geq 0} \{d_{\alpha_2 + \gamma_{r_1,0}, \beta}(t) \} 
        = h(\alpha_2 + \gamma_{r_1,0}, \beta).
        \end{align*}
        
        Using the piecewise linear nature of the curves, it is clear that this horizontal deviation is taken on at $\tau = a^*_{\alpha_2 + \gamma_{r_1,0},\beta}$, and thus
        \begin{align*}
        \theta_{\text{opt}} = 
        h(\alpha_2+\gamma_{r_1,0}, \beta)=
        \frac{b^{\tau}_2+(r_2^{\tau}+r_1)\cdot \tau}{R^{\beta^{-1}({\alpha_2(\tau))}}} + T^{\beta^{-1}({\alpha_2(\tau))}} - \tau . \hspace*{20mm} \qed
        \end{align*}} \end{proof}
\end{thm}

Using Thm.~\ref{thm:basecaseNew}, we can find the optimal value $\topt^i$ that minimizes $\qmin^i$ for each $\gamma_i$ corresponding to segment $S_i$ of $\foi$. Considering Eq.~\eqref{eq:concp2}, we recognize that $\qmin^i$ may not be assumed at a valid point in time when considering $\foi$ as a whole. Indeed, each segment $S_i$ is only defined over its respective interval $I_i\coloneqq[a_i, a_{i+1})$. If now $\topt^i\notin I_i$, we adjust its value as follows (with an arbitrarily small $\epsilon>0$):
    \begin{equation}\label{eq:tqadjust}
    \toptp^i \coloneqq \begin{cases}
        a_i, & \topt^i < a_i, \\
        a_{i+1}-\epsilon, & \topt^i > a_{i+1}, \\
        \topt^i, & \textrm{otherwise}.
    \end{cases}
    \end{equation}
This adjustment follows from the properties of the involved curves. We know that $\gamma_i$ and $\cross$ are concave curves, their sum is hence also concave. In addition, the service curve $\beta$ is convex. 
If we now assume $\topt^i$ outside of the interval $I_i$, we can use these properties to find the largest backlog bound that is still assumed over $I_i$ as follows: In case 1 of Eq.~\eqref{eq:tqadjust}, if $\topt^i < a_i$, we set $\toptp^i = a_i$.
Since $\foi$ and $\cross$ are concave, their aggregate rate decreases in time, while due to the convexity of $\beta$, its rate increases in time. As such, measuring the backlog at any point $x\in I_i$, $x>a_i$, results in a smaller backlog than if measured at $a_i$. For case 2 it is similar, but we need to take the half-open nature of $I_i$ into account. Using an arbitrarily small $\epsilon>0$, and assuming the backlog at $a_{i+1}-\epsilon$ ensures that we can approximate the largest possible backlog that is still in $I_i$ as close as we desire.
This leaves us with two vectors of $n$ values for $\topt^i$ and $\toptp^i$. For purposes of the heuristic, we let $\init\coloneqq [\topt^1, \dots, \topt^n]$ and $\post\coloneqq[\toptp^1,\dots,\toptp^n]$. We proceed with two observations about $\topt^i$ and $\toptp^i$:

\begin{lem}\label{cor:tqiorder}
It holds that $\topt^1\geq \dots \geq \topt^n$.
\end{lem}

\begin{lem}\label{prop:uniqueequal} 
There is at most one $\topt^i$ with $\topt^i=\toptp^i$.
\end{lem}

Using these insights, we proceed with our decomposition heuristic:

\begin{itemize}
    \item[1.] Calculate $\diff\coloneqq \init - \post = [\topt^1 - \toptp^1, \dots, \topt^n - \toptp^n]$.
    \item[2.] Check whether there exists an entry with value $i=0$ in $\diff$. If yes, \emph{return $v(\foi, \beta_{\topt^i}^1)$ and break.}
    \item[3.] If no such entry exists in $\diff$, iterate over $t\in A_1$ and find the first intersection $\phi$. \emph{Return $v(\foi, \beta_\phi^1)$ and break.}
\end{itemize}

Let us discuss the rationale and potential shortcomings of the heuristic:
If a $\topt^i$ is unchanged after applying Eq.~\eqref{eq:tqadjust}, it may seem that we should have obtained the same value of $\topt$ when using Thm.~\ref{thm:argmin_of_max_is_opt_theta} on $\foi$ without decomposing the curve into its $n$ segments. Consequently, setting $\topt = \topt^i$ seems to be the correct choice in this case. From Lem.~\ref{prop:uniqueequal}, we know that the heuristic can find at most one such $\topt^i$, hence the result is unique. However, in fact, this does not yield the same result as using Thm.~\ref{thm:argmin_of_max_is_opt_theta} in all cases. Whenever $\foi(\astar)$ is assumed at a different segment than $\foi(\astar+\topt)$, the heuristic calculates a larger $\topt$, and consequently larger backlog bound, as it only considers one segment of $\foi$. 
Furthermore, the heuristic can also falsely assume that there exists no entry $i=0$ in $\diff$. In step 3, we know that all $\toptp^i$ in $\diff$ are $t\in A_1$, as we have adjusted them to their respective interval limits using Eq.~\eqref{eq:tqadjust}\footnote{They could be off by $\epsilon$, but we can let $\epsilon \rightarrow 0$.}. For each $\toptp^i$, we calculate its $v_{\toptp^i}$ curve and obtain the intersection with $\foi$. From Rem.~\ref{rem:first_intersections_is_last_intersection}, it follows that the largest of these intersections is equal to $\phi$, hence the $\topt$ that minimizes the backlog bound. We remark that the heuristic always returns a backlog bound, as it always finds a value for $\theta$ -- it may, however, be conservative.

\subsection{Evaluation of the Decomposition Heuristic} \label{sub:heueval} 

We continue with an evaluation of the decomposition heuristic with respect to its accuracy and efficiency. To this end, we calculate and compare the backlog bounds and execution times of the exact method using Thm.~\ref{thm:argmin_of_max_is_opt_theta} and the decomposition heuristic. We consider the following scenario: Assume we have a variable number of 2 to 10 cross flows at the node, multiplexed in a FIFO aggregate with the foi. Let $\foi$ be a PWL concave arrival curve for the foi with either two or four segments (as typically found in literature, e.g. \cite{wrege1996deterministic}). Each cross flow arrival curve $\cross^i$ is modeled as a T-Spec curve, defined as $\cross^i(t)\coloneqq \min\{b_2^1 + r_2^1t, b_2^2+r_2^2t\}$. We pick the packet size $b_i^1$ for each flow $i$ randomly from \qtyrange{0.001}{0.05}{\mega\bit}, the sustained rate $r_i^2$ from \qtyrange{1}{10}{Mbit/s}, and and the first breakpoint $a_1^i$ from \qtyrange{50}{500}{\milli\second}. For $\foi$, we pick a random breakpoint spacing from \qtyrange{100}{500}{\milli\second} that is added to $a_1^1$ to obtain further breakpoints. We set the peak rate for each arrival curve as $r_i^1 = r_i^2 \cdot 8$, similar to \cite{le2001network}. For $\foi$ with four segments, we set the rate of the fourth segment as sustained rate, and assign rate $6\cdot r_i^1$ to segment two, and rate $3\cdot r_i^1$ to segment three. The burst value of each subsequent token bucket is set as
\begin{equation}
    b^i=b^{i-1} - (r^i - r^{i-1})\cdot a_i.
\end{equation}
The server offers a rate-latency service curve $\rl$, where the rate is set such that there is an 80\% utilization of the server when considering the sum of sustained rates of all flows. The latency is set as $T=1/R$ \s{}.  

The experiments were run on hardware using an Intel i7-1165G7 CPU with 4 cores at \ghz{2.8} base frequency, \gb{32} DDR4-\mhz{3200} RAM, a \gb{512} NVMe SSD and the Ubuntu 22.04 LTS operating system. The experiments are implemented using the Nancy library \cite{zippo2022nancy} for C\#, and are run in JetBrains Rider 2024.3 with the .NET SDK 9.0.102 using default compiler settings. 

We perform 500 iterations, where in each one of them we iterate again over the number of cross flows from 2 to 10. For each run, we measure the execution time using \emph{System.Diagnostics.Process.UserProcessorTime}. This method returns the CPU time of the process that is running the experiments, without including any other processing time. We only measure code fragments that are used to calculate one of the two methods, ignoring any other code parts. We take the execution time of both methods and calculate the percentage difference of both measured times. Of course, we also compare the backlog bounds obtained by using the exact method and the heuristic.

\begin{table}
    \centering
    \begin{tabular}{|c|c|c|c|c|c|c|c|c|c|}
      \hspace{2px}Cross [\#] & 2 & 3 & 4 & 5 & 6 & 7 & 8 & 9 & 10 \\ \hline \hline 
        $\muex$ & 15.1 & 15.7 & 16.1 & 16.4 & 16.6 & 16.7 & 16.9 & 17.0 & 17.1 \\
        CI(95\%) &  0.8 & 0.9 & 0.9 & 0.9 & 0.9 & 0.9 & 0.9 & 0.9 & 0.9  \\ \hline 
        $\muheu$ & 15.1 & 15.7 & 16.1 & 16.4 & 16.6 & 16.7 & 16.9 & 17.0 & 17.1 \\
        CI(95\%) & 0.8 & 0.9 & 0.9 & 0.9 & 0.9 & 0.9 & 0.9 & 0.9 & 0.9 \\ \hline
        Acc. [\%] & 98.4 & 99.6 & 99.6 & 100 & 100 & 100 & 100 & 100 & 100   \\ \hline 
        Inc. [\%] & 1.6 & 6.7 & 3.3 & 0 & 0 & 0 & 0 & 0 & 0     \\ 
        CI(95\%) & 0.0 & 0.3 & 0.1 & 0 & 0 & 0 & 0 & 0 & 0     \\ \hline \hline 
        $\execex$ [ms] & 976 & 1362 & 1742 & 2114 & 2520 & 2918 & 3335 & 3758 & 4202 \\ 
        $\execheu$ [ms] & 14.3 & 12.6 & 12.8 & 13.0 & 13.3 & 13.6 & 14.1 & 14.2 & 14.6   \\ \hline 
        Speedup & 68 & 108 & 136 & 163 & 189 & 215 & 237 & 265 & 288 \\ \hline 
    \end{tabular}
    \caption{Comparison of exact method and heuristic for $\foi$ with two segments.}\label{tab:occurences}
\end{table}

The results of the experiments for two segments are shown in Table~\ref{tab:occurences}. The results for four segments can be found in Appendix~D. 
The tables are organized as follows. We report the mean and confidence interval (CI) for the backlog bound ($\muex$ and $\muheu$) as well as execution times ($\execex$ and $\execheu$) for each number of cross flows. The accuracy reports how often the resulting backlog bound using the heuristic is equal to the backlog bound of the exact method. We also calculate the percentage increase and its CI from the exact method to the heuristic in cases where the results are not identical. For the execution times, the speedup from the exact method to the heuristic is calculated as $\execex / \execheu$. Furthermore, for 3 and 7 cross flows, we illustrate the relation of the obtained backlog bounds and execution times for both methods and segment counts in Figs~\ref{fig:bl} and \ref{fig:exec}. For both figures, the scenario numbers are ordered by the value of the respective heuristic results (execution times and backlog bounds), from smallest to largest. For the backlog bounds, whenever the exact method calculates a more accurate backlog bound, this is represented as a spike down in the plot.

\begin{figure}
\centering
\includegraphics[scale=0.54]{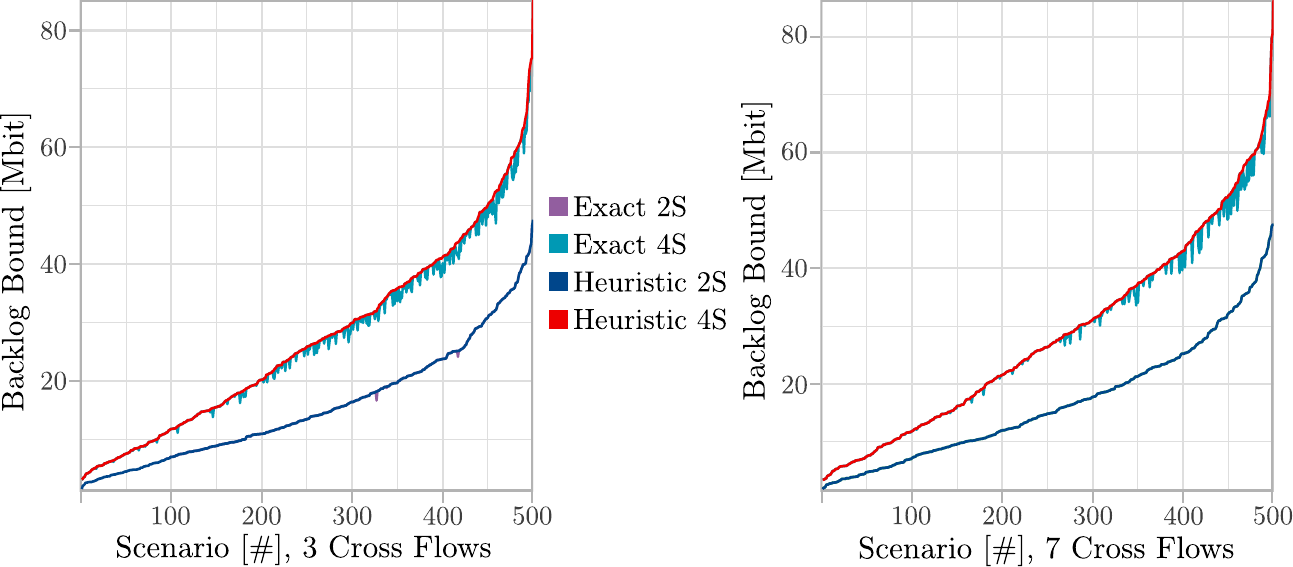}
\caption{Backlog bounds for varying numbers of cross flows and foi segments.}\label{fig:bl}
\end{figure}

\begin{figure}
\centering
\includegraphics[scale=0.52]{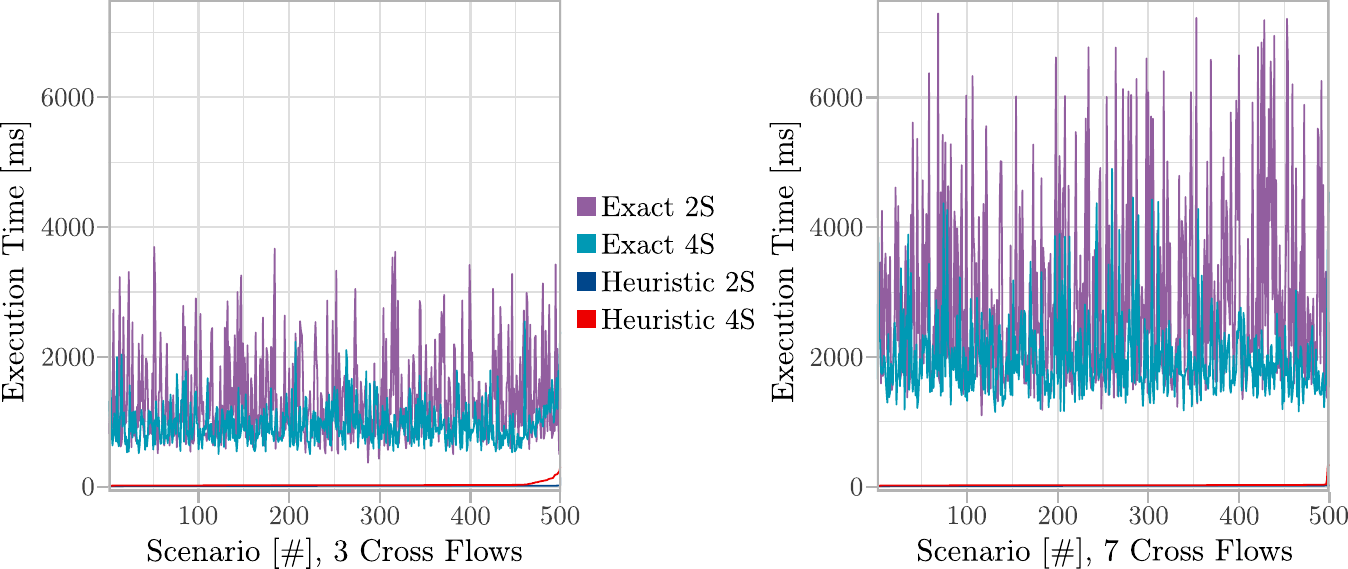}
\caption{Execution time for varying numbers of cross flows and foi segments.}\label{fig:exec}
\end{figure}
\vspace*{-5mm}
We observe that, for two segments, the heuristic provides very accurate results compared to the exact method. For 2 to 4 cross flows, a small percentage of runs do not produce the same backlog bound. For these, we observe a percentage increase of up to 6.7\% of the backlog bound for 3 cross flows. (Yet, the overall mean of the iterations is not affected by this.) The accuracy for four $\foi$ segments is worse, but the percentage increase stays low across all numbers of cross flows. For the execution times, we observe that it grows in the number of cross flows. Additionally, in Table~\ref{tab:occurences} we observe that the ratio of the execution times (speedup) between the two methods grows significantly. For four segments, we observe the same behavior, with a slower speedup across cross flow numbers. 

\section{Usage of the Results in the DiscoDNC Tool}\label{sec:dblanalysis}

In the previous section, we have conducted our experiments using Nancy. Nancy is a general toolbox to support network calculus analyses but does not provide complete analysis methods by itself. Consequently, it does not mandate any specific value for $\theta$, but instead we used it to compute such values. In contrast to Nancy, another open-source tool called DiscoDNC \cite{schmitt2006disco,DiscoDNCv2,scheffler2021network} provides complete network calculus analyses, in particular also for FIFO with a default setting of $\theta$ as follows:  $\discot=\beta^{-1}\left(\sum_{i=2}^n b_1^i\right)$.
In fact, this $\theta$ value coincides with the one that was used in Section~\ref{sec:backlog_exact} as a lower bound to the optimal theta value $\topt$ \emph{when} all arrival curves are token buckets and the service curve is a rate-latency service curve. For more complex curves, it holds that $\discot < h(\alpha_2,\beta)$, however. Setting $\theta \leq h(\alpha_2,\beta)$ results in a $\lo$ that is continuous in all points. It was consequently a reasonable first choice for an NC tool.

Yet, now we are able to adjust the $\theta$ value employed by DiscoDNC using our exact method and the decomposition heuristic. In the following, we are interested in the increase in quality of the calculated backlog bounds we can achieve when adjusting $\theta$ using the exact method as well as the heuristic. 

\begin{figure}
\begin{subfigure}{0.495\textwidth}
    \centering
    \includegraphics[width=0.95\textwidth]{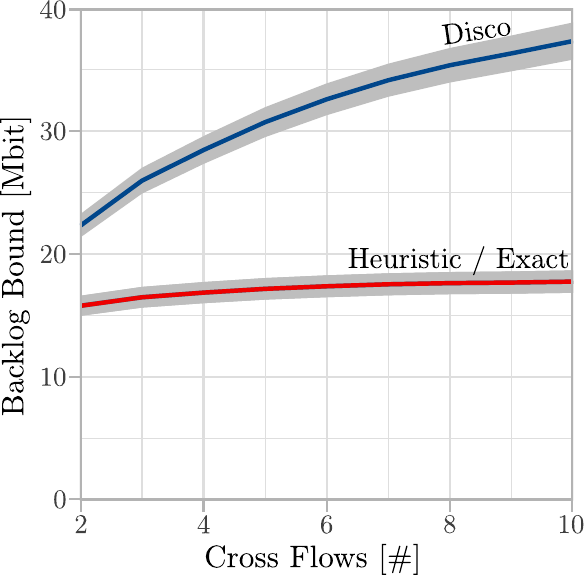}
    \caption{Mean backlog bound.}\label{fig:discocomp}
\end{subfigure}
\begin{subfigure}{0.495\textwidth}
\centering
\includegraphics[width=0.95\textwidth]{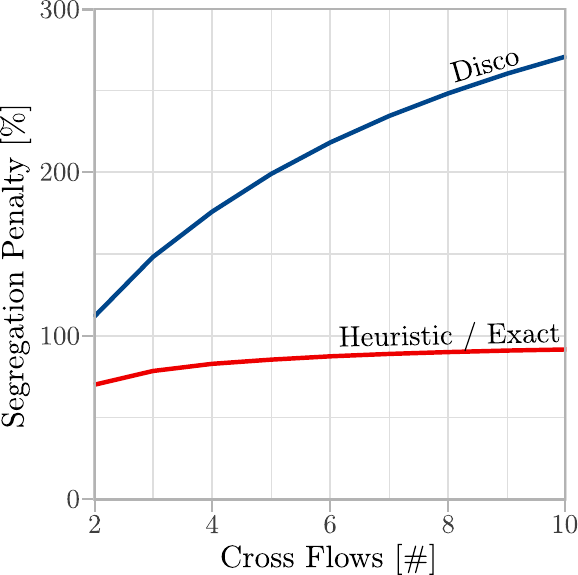}
\caption{Segregation penalty.}\label{fig:segpen}
\end{subfigure}
    \caption{Comparisons of exact method and heuristic against the DiscoDNC tool.}
    \label{fig:discobacklog}
\end{figure}

We use the same experiment setup that we have used in Section~\ref{sub:heueval}. We compare the obtained backlog bound when using $\discot$ against the backlog bound obtained using the exact method and heuristic. The results are illustrated in Fig.~\ref{fig:discocomp}. Here, we have calculated the mean backlog bounds and confidence intervals (CI) for the three methods for different numbers of cross flows. We observe that for all numbers of cross flows, the exact method and heuristic produce very similar backlog bounds (invisible difference in the figure). For $n=10$, the backlog bound of both methods has the same mean and CI, with CI(95\%)$=17.76\pm 0.308$. The backlog bound for the default value $\discot$ has CI(95\%)$=37.34\pm 0.506$. For two cross flows, the exact method and heuristic have CI(95\%)$=15.8\pm 0.272$, and the default value has CI(95\%)$=22.3\pm 0.320$. 

Additionally, we study the segregation penalty (see Section~\ref{sec:sysmod}) for this experimental setup. To this end, we calculate the segregation penalty as $(\sum \qmin^i - q_{\mathrm{agg}}) / q_{\mathrm{agg}} \cdot 100$, where $q_{\mathrm{agg}}$ is the aggregate backlog bound. The results are given in Fig.~\ref{fig:segpen}. We observe that the segregation penalty using the default $\theta$ value of DiscoDNC grows significantly in the number of cross flows, in contrast to the new methods. This may potentially result in sub-optimal system design decisions.

To conclude, the backlog bound using $\discot$ is significantly less accurate than the other two methods for any number of cross flows, with the difference becoming larger and larger the more cross flows traverse the server.

\section{Conclusion}

In this paper, we presented an exact method for deriving FIFO residual service curves that minimize backlog bounds for PWL concave arrival curves. While this method provides precise results, its computational inefficiency becomes apparent in more complex scenarios, such as settings with a large number of crossflows. 
To address this limitation, we introduced an efficient heuristic that provides an upper bound on the backlog-minimizing parameter of the FIFO residual service curve. Despite its approximate nature, we could show the heuristic yields backlog bounds that are close to those of the exact method, particularly in scenarios with a higher number of crossflows. Importantly, it achieves this with significantly reduced execution time.
In a final experiment with the DiscoDNC tool, both the exact and heuristic approaches produced significantly more precise backlog bounds than the existing default setting of the tool.

\bibliographystyle{splncs04}
\bibliography{References}

\begin{thebibliography}{10}
\providecommand{\url}[1]{\texttt{#1}}
\providecommand{\urlprefix}{URL }
\providecommand{\doi}[1]{https://doi.org/#1}

\bibitem{blanc2006quality}
Blanc, A.P.: Quality of service guarantees for {FIFO} queues with constrained
  inputs. University of California, San Diego (2006)

\bibitem{DiscoDNCv2}
Bondorf, S., Schmitt, J.B.: The {DiscoDNC} v2 -- a comprehensive tool for
  deterministic network calculus. In: Proc. of the International Conference on
  Performance Evaluation Methodologies and Tools. pp. 44--49. ValueTools '14
  (December 2014), \url{https://dl.acm.org/citation.cfm?id=2747659}

\bibitem{bouillard2018deterministic}
Bouillard, A., Boyer, M., Le~Corronc, E.: Deterministic Network Calculus: From
  Theory to Practical Implementation. John Wiley \& Sons (2018)

\bibitem{cholvi2002worst}
Cholvi, V., Echag{\"u}e, J., Le~Boudec, J.Y.: Worst case burstiness increase
  due to {FIFO} multiplexing. Performance Evaluation  \textbf{49}(1-4),
  491--506 (2002)

\bibitem{cruz1991calc1}
Cruz, R.L.: A calculus for network delay. {I}. {N}etwork elements in isolation.
  IEEE Transactions on Information Theory  \textbf{37}(1),  114--131 (1991)

\bibitem{cruz1998sced+}
Cruz, R.L.: {SCED}+: Efficient management of quality of service guarantees. In:
  Proceedings. IEEE INFOCOM'98, the Conference on Computer Communications.
  Seventeenth Annual Joint Conference of the IEEE Computer and Communications
  Societies. Gateway to the 21st Century (Cat. No. 98. vol.~2, pp. 625--634.
  IEEE (1998)

\bibitem{le2001network}
Le~Boudec, J.Y., Thiran, P.: Network Calculus: A Theory of Deterministic
  Queuing Systems for the Internet. Springer (2001),
  \url{https://leboudec.github.io/netcal/}

\bibitem{lenzini2005delay}
Lenzini, L., Mingozzi, E., Stea, G.: Delay bounds for {FIFO} aggregates: a case
  study. Computer Communications  \textbf{28}(3),  287--299 (2005)

\bibitem{liebeherr2009system}
Liebeherr, J., Fidler, M., Valaee, S.: A system-theoretic approach to bandwidth
  estimation. IEEE/ACM Transactions on networking  \textbf{18}(4),  1040--1053
  (2009)

\bibitem{mckeown2002achieving}
McKeown, N., Mekkittikul, A., Anantharam, V., Walrand, J.: Achieving 100\%
  throughput in an input-queued switch. IEEE Transactions on Communications
  \textbf{47}(8),  1260--1267 (2002)

\bibitem{scheffler2021network}
Scheffler, A., Bondorf, S.: Network calculus for bounding delays in feedforward
  networks of {FIFO} queueing systems. In: Proc. of the 18th International
  Conference on Quantitative Evaluation of Systems. pp. 149--167. QEST '21
  (August 2021),
  \url{https://link.springer.com/chapter/10.1007/978-3-030-85172-9_8}

\bibitem{schmitt2006disco}
Schmitt, J.B., Zdarsky, F.A.: The disco network calculator: a toolbox for worst
  case analysis. In: Proceedings of the 1st international conference on
  Performance evaluation methodologies and tools. pp. 8--es (2006)

\bibitem{shvachko2010hadoop}
Shvachko, K., Kuang, H., Radia, S., Chansler, R.: The hadoop distributed file
  system. In: 2010 IEEE 26th symposium on mass storage systems and technologies
  (MSST). pp. 1--10. Ieee (2010)

\bibitem{wrege1996deterministic}
Wrege, D.E., Knightly, E.W., Zhang, H., Liebeherr, J.: Deterministic delay
  bounds for {VBR} video in packet-switching networks: fundamental limits and
  practical trade-offs. IEEE/ACM Transactions on networking  \textbf{4}(3),
  352--362 (1996)

\bibitem{zippo2022nancy}
Zippo, R., Stea, G.: Nancy: an efficient parallel network calculus library.
  SoftwareX  \textbf{19},  101178 (2022)

\end{thebibliography}

\newpage 
\appendix

\section{Proofs for Section 3 (System Model)}

\setcounter{thm}{13}
\begin{lem} 
    Let $\alpha \in \mathcal{F}$ and $\beta \in \mathcal{F}$ be given. Then it holds that
    \begin{align}
        v(\alpha,\beta) = v(\alpha, \beta_\downarrow). 
    \end{align}

    \begin{proof} \normalfont{ 
        \begin{align*}
            v(\alpha, \beta_\downarrow) &= \sup_{t \geq 0} \{\alpha(t) - \beta_\downarrow(t) \} \\
            &= \sup_{t \geq 0} \{\alpha(t) - \inf_{s \geq 0}\{ \beta(t+s) \} \} \\
            &= \sup_{t \geq 0} \sup_{s \geq 0} \{ \alpha(t) - \beta(t+s) \} \\
            &= \sup_{s \geq 0} \sup_{t \geq 0} \{ \alpha(t) - \beta(t+s) \} \\
            &= \sup_{s \geq 0} \ v(\alpha, \beta \otimes \delta_{-s}) \\
            &= \sup_{s \geq 0} \ v(\alpha \otimes \delta_{s}, \beta) \\
            \overset{s=0}&{=} v(\alpha,\beta)
        \end{align*}
    } \qed \end{proof}
\end{lem}

\section{Proofs for Section 4 (Derivation of the $\theta$ Parameter for Minimal Per-Flow Backlog Bounds)}

\begin{lem} 
    Let $\alpha_1$ and $\alpha_2$ be PWL concave arrival curves and let $\beta$ be a PWL convex service curve under FIFO multiplexing. Let $h(\alpha_2,\beta)$ be the horizontal deviation of $\alpha_2$ and $\beta$. Then, it holds for $0 \leq \theta \leq h(\alpha_2,\beta)$ that
    \begin{align*}
        v(\alpha_1, \beta_{\theta}^1) \geq v(\alpha_1, \beta_{h(\alpha_2,\beta)}^1).
    \end{align*}
    \begin{proof} \normalfont{
        Let $0 \leq \theta \leq h(\alpha_2, \beta)$. Clearly, it holds that
        \begin{align}
            \alpha_2(t-\theta) \geq \alpha_2(t-h(\alpha_2, \beta)). \label{Eq:a2_of_theta_geq_a2_of_h}
        \end{align}
        Let $z(\theta)$ be a function of the last intersection of $\alpha_2(t-\theta)$ and $\beta(t)$, given by 
        \begin{align*}
            z(\theta) = \sup\{ t \geq 0 \ | \ \alpha_2(t-\theta) = \beta(t) \}.
        \end{align*}
        It holds that $z(\theta) \geq h(\alpha_2,\beta) \geq \theta$ and that $z(\theta)$ is decreasing for increasing $\theta$. So $z(\theta)$ is minimal for $\theta=h(\alpha_2,\beta)$. Note that we can write $\beta_{\theta \downarrow}^1=\lo\cdot \mathbbm{1}_{\{t \geq z(\theta) \}}$ as $\beta_{\theta}^1(t)$ is non-decreasing for $t\geq z(\theta)$.
        Then we obtain
        \begin{align*}
            v(\alpha_1, \lo) \overset{Eq.(\ref{Eq:vertical_dev_with_LNDC})}&{=} v(\alpha, \beta_{\theta \downarrow}^1 ) \\
            &= \sup_{t\geq0} \{ \alpha_1(t) - \beta_{\theta \downarrow}^1(t) \} \\
            &= \sup_{t\geq0} \{ \alpha_1(t) - ([\beta(t) - \alpha_2(t-\theta)]^+ \cdot \mathbbm{1}_{\{t \geq \theta \}} ) \cdot \mathbbm{1}_{\{t \geq z(\theta) \}} \} \\
            \overset{z(\theta) \geq \theta}&{=} \sup_{t\geq0} \{ \alpha_1(t) - [\beta(t) - \alpha_2(t-\theta)]^+ \cdot \mathbbm{1}_{\{t \geq z(\theta) \}}  \} \\
            \overset{Eq.(\ref{Eq:a2_of_theta_geq_a2_of_h}), \theta\leq h(\alpha_2,\beta)}&{\geq}  \sup_{t\geq0} \{ \alpha_1(t) - [\beta(t) - \alpha_2(t-h(\alpha_2,\beta))]^+ \cdot \mathbbm{1}_{\{t \geq z(h(\alpha_2,\beta)) \}}  \} \\
            \overset{z(\theta) \geq h(\alpha_2,\beta)}&{=}  \sup_{t\geq0} \{ \alpha_1(t) - [\beta(t) - \alpha_2(t-h(\alpha_2,\beta))]^+
            \cdot \mathbbm{1}_{\{t \geq h(\alpha_2,\beta) \}} 
            \cdot \mathbbm{1}_{\{t \geq z(h(\alpha_2,\beta)) \}} \}  \\
            &= v(\alpha, \beta_{h(\alpha_2,\beta) \downarrow}^1 ) \overset{Eq.(\ref{Eq:vertical_dev_with_LNDC})}{=} v(\alpha, \beta_{h(\alpha_2,\beta)}^1 )
        \end{align*}
    } \qed \end{proof}
\end{lem}

\begin{lem} 
    Let $\alpha$ be a PWL concave function and $\beta$ be a PWL convex function. Let $A$ and $B$ be the set of breakpoints of $\alpha$ and $\beta$, respectively. Then, the vertical deviation of $\alpha$ and $\beta$ can always be calculated at some time $t \in A \cup B$.

    \begin{proof} \normalfont{ 
        Suppose the vertical deviation $v(\alpha,\beta)$ of $\alpha$ and $\beta$ is assumed or calculated at $t' \in \mathbb{R} \setminus (A \cup B)$ and cannot be calculated at $t \in A \cup B$. So $v(\alpha,\beta) = \sup_{\tau \geq 0}\{\alpha(\tau) - \beta(\tau)\} \overset{\tau=t'}{=} \alpha(t) - \beta(t)$ holds. 
        Since the stability condition holds and $\alpha$ is concave and $\beta$ is convex, there always exists such a $t'$ for which we assume the supremum. 
        
        Consider the linear segments of $\alpha$ and $\beta$ at time $t'$, in particular the respective rates $r^{t'}$ and $R^{t'}$ (according to Def.~\ref{def:linear_segment_at_time_t}).
        There are three cases to consider: $r^{t'} = R^{t'}$, $r^{t'} > R^{t'}$ and $r^{t'} < R^{t'}$. \\
        \underline{Case I} ($r^{t'} = R^{t'}$): Since $r^{t'} = R^{t'}$ holds the vertical distance of $\alpha$ and $\beta$ does neither increase nor decrease for the linear segments $\alpha^{t'}$ and $\beta^{t'}$. So for breakpoints $t_l = \max \{t \in A \cup B : t < t'\}$ and $t_u = \min \{t \in A \cup B : t > t'\}$ it should also hold that $v(\alpha,\beta) = \alpha(t') - \beta(t') = \alpha(t_l) - \beta(t_l) = \alpha(t_u) - \beta(t_u)$. $\quad \lightning$ \\
        \underline{Case II} ($r^{t'} > R^{t'}$): Since $r^{t'} > R^{t'}$ holds, the vertical distance of $\alpha$ and $\beta$ is increasing with time until the first breakpoint for which the rate of $\alpha$ becomes lower than the rate of $\beta$. Let $x = a^*_{\alpha,\beta}$ (see Def.~\ref{def:a_star_and_s_star}). So, it should also hold that $v(\alpha,\beta) = \alpha(t') - \beta(t') < \alpha(x) - \beta(x). \quad \lightning$  \\
        \underline{Case III} ($r^{t'} < R^{t'}$): Analogous to Case II.
    } \qed \end{proof}
\end{lem}

\setcounter{thm}{17}
\begin{lem} 
    Let $\max_{t \in \mathcal{T}} \{v_t(\theta)\}$ and $\alpha_1(\theta)$ be defined as above. For the endpoints of the given interval $[h(\alpha_2,\beta), t_{\max}]$ it holds that
    \begin{align*}
        \alpha_1(h(\alpha_2,\beta)) \leq \max_{t \in \mathcal{T}} \{v_t(h(\alpha_2,\beta))\} \ \textrm{and} \
        \alpha_1(t_{\max}) = \max_{t \in \mathcal{T}} \{v_t(t_{\max})\}.
    \end{align*}

    \begin{proof}
        \normalfont{
            By the definition of $v_t(\theta)$ curves it holds that $\max_{t \in \mathcal{T}} \{v_t(\theta)\} \geq \alpha_1(\theta)$, since $\alpha_1(\theta) = v_0(\theta) \in \{v_t(\theta) \ | \ t \in \mathcal{T} \}$ for $\theta \geq h(\alpha_2, \beta) > 0$. This directly implies that $\alpha_1(h(\alpha_2,\beta)) \leq \max_{t \in \mathcal{T}} \{v_t(h(\alpha_2,\beta))\}$.
            We also know by the definition of $v_t(\theta)$ curves that $\max_{t \in \mathcal{T}} \{v_t(t_{\max})\} = \alpha_1(t_{\max})$ holds.
        } \qed 
    \end{proof}
\end{lem}

\section{Proofs for Section 5 (Efficient Calculation of Near-Optimal Backlog Bounds)}

\setcounter{thm}{21}
\begin{lem}
It holds that $\topt^1\geq \dots \geq \topt^n$.
\end{lem}
\begin{proof}
From Eq.~\eqref{eq:concp1}, we know that the rate of $\gamma_i$ decreases in $i$. Together with $\beta_\theta^1$ being convex, the largest time at which the backlog of a single segment of $\foi$ is assumed is for the first segment $\topt^1$, as $S_1$ has the highest rate of all segment curves. Consequently, all subsequent $\topt^i$ are ordered and decrease in $i$. 
\qed \end{proof}
\begin{lem}
There is at most one $\topt^i$ with $\topt^i=\toptp^i$.
\end{lem}
\begin{proof}
Let an arbitrary but fixed $\topt^i=\toptp^i$, i.e., $\topt^i\in I_i$ . Then, using Lem.~\ref{cor:tqiorder}, we know that $\topt^1,\dots,\topt^{i-1}\geq \topt^i$. Since $\topt^i\in I_i$, no earlier segment can assume its backlog on its respective interval according to Eq.~\eqref{eq:concp2}, as the respective interval limits are $< a_i$. Likewise, $\topt^{i+1},\dots,\topt^n\leq \topt^i$, i.e., no later segment can assume its backlog on its respective minimal interval as all interval limits $\geq a_{i+1}$. It follows that if $\exists \topt^i, \topt^i=\toptp^i$, there exists exactly one such entry. 
\qed \end{proof}

\section{Table for Section 5 (Efficient Calculation of Near-Optimal Backlog Bounds)}

\begin{table}
    \centering
    \begin{tabular}{|c|c|c|c|c|c|c|c|c|c|}
      \hspace{2px}Cross [\#] & 2 & 3 & 4 & 5 & 6 & 7 & 8 & 9 & 10 \\ \hline \hline 
        $\muex$ & 26.1 & 26.9 & 27.4 & 27.7 & 28.0 & 28.2 & 28.4 & 28.5 & 28.6 \\
        CI(95\%) & 1.3 & 1.4 & 1.4 & 1.5 & 1.5 & 1.4 & 1.4 & 1.5 & 1.5   \\ \hline 
        $\muheu$ & 26.6 & 27.4 & 27.8 & 28.2 & 28.4 & 28.5 & 28.6 & 28.8 & 28.8 \\
        CI(95\%) & 1.4 & 1.4 & 1.5 & 1.4 & 1.4 & 1.5 & 1.5 & 1.5 & 1.5 \\ \hline
        Acc. [\%] & 67.6 & 73.0 & 77.2 & 79.4 & 81.6 & 84.6 & 85.6 & 86.8 & 89.0   \\ \hline 
        Inc. [\%] & 4.3 & 4.8 & 5.5 & 5.2 & 4.5 & 4.9 & 4.7 & 4.0 & 3.7   \\ 
        CI(95\%) & 0.3 & 0.2 & 0.2 & 0.3 & 0.2 & 0.2 & 0.2 & 0.3 & 0.2    \\ \hline \hline 
        $\execex$ [ms] & 766 & 971 & 1201 & 1457 & 1714 & 1993 & 2253 & 2527 & 2799 \\ 
        $\execheu$ [ms] & 33 & 29 & 26 & 25 & 26 & 26 & 24 & 25 & 25 \\
        Speedup & 23 & 34 & 47 & 59 & 67 & 78 & 93 & 101 & 110 \\ \hline 
    \end{tabular}
    \caption{Comparison of exact method and heuristic for $\foi$ with four segments.}
\end{table}

\end{document}